\documentclass{ceurart}

\newcommand{\id}{{\tt id}}
\newcommand{\source}{{\tt s}}
\newcommand{\target}{{\tt t}}
\newcommand{\Monoid}{{\tt M}}
\newcommand{\ToffoliM}{{\tt t}}
\newcommand{\leftM}{{\tt l}}
\newcommand{\rightM}{{\tt r}}

\newcommand{\N}{{\mathbb{N}}}

\newcommand{\setC}{{\mathcal C}}
\newcommand{\setD}{{\mathcal D}}

\newcommand{\ToffoliCategory}{ {\scalebox{.9}[1.0]{\bf Toff}} }
\newcommand{\Move}{ {\scalebox{.9}[1.0]{\bf Move}} }

\usepackage[dvipsnames]{xcolor} 
\definecolor{citegreen}{rgb}{0,0.5,0.2} 
\definecolor{codegreen}{rgb}{0,0.6,0}

\usepackage{hyperref} 
\hypersetup{ 
    colorlinks=true,
    linkcolor=blue,
    filecolor=magenta,
    urlcolor=cyan,
    citecolor=citegreen,
    bookmarks=true,
}
\usepackage{amsmath}
\usepackage{amssymb}
\usepackage{amsthm}

\usepackage{stmaryrd}
\usepackage{yhmath}
\usepackage{mathtools}
\usepackage{extarrows}
\usepackage{enumerate}
\usepackage{lineno}
\usepackage{xcolor}

\usepackage{float}
\floatstyle{boxed}
\restylefloat{figure}
\usepackage{caption}
\usepackage{subcaption}

\usepackage{catex}
\usepackage{tikz}
\usetikzlibrary{cd, decorations.markings}

\usepackage{tikz-cd}

\tikzset{
Rightarrow/.style={double equal sign distance,>={Implies},->},
duble/.style={draw, -{Classical TikZ Rightarrow[length=1mm]}, double distance=1pt},
triple/.style={-,preaction={draw,Rightarrow}},
quadruple/.style={preaction={draw,Rightarrow,shorten >=0pt},shorten >=1pt,-,double,double distance=0.2pt},
labelarrow/.style={postaction=decorate,decoration={markings,mark=at position 0.5 with {\arrow{<};}}}
}

\newtheorem{proposition}{Proposition}
\newtheorem{theorem}{Theorem}

\newtheorem{lemma}{Lemma}

\theoremstyle{definition}
\newtheorem{definition}{Definition}

\theoremstyle{remark}
\newtheorem{remark}{Remark}

\begin{document}

\deftwocell[dots]{n : 2 -> 2} 
\deftwocell[text = f]{f : 2 -> 2}
\deftwocell[text = g]{g : 2 -> 2}
\deftwocell[text = h]{h : 2 -> 2}
\deftwocell[text = k]{k : 2 -> 2}
\deftwocell[rectangle]{ex1 : 1 -> 1}
\deftwocell[rectangle]{ex2 : 2 -> 2}
\deftwocell[rectangle]{ex3 : 3 -> 3}
\deftwocell[rectangle]{gen : 2 -> 2}
\deftwocell[left = \overline{x}]{xl : 2 -> 2}
\deftwocell[left = \overline{y}]{yl : 2 -> 2}
\deftwocell[text = \mathtt{N}]{not : 1 -> 1}
\deftwocell[text = \mathtt{T}_{2}]{t2 : 2 -> 2}
\deftwocell[text = \mathtt{T}_{3}]{t3 : 3 -> 3}
\deftwocell[crossing]{sw : 2 -> 2}
\deftwocell[crossing1]{lad1 : 3 -> 3}
\deftwocell[crossing2]{lad1' : 3 -> 3}
\deftwocell[crossing1]{lad2 : 4 -> 4}
\deftwocell[crossing2]{lad2' : 4 -> 4}
\deftwocell[mid = x]{wx : 1 -> 1}
\deftwocell[mid = y]{wy : 1 -> 1}
\deftwocell[mid = c]{wc : 1 -> 1}
\deftwocell[mid = 1-x]{oN : 1 -> 1}
\deftwocell[right = \:\:\:\:\:\:c + x]{oT2 : 1 -> 1}
\deftwocell[right = \:\:\:\:\:\:c + xy]{oT3 : 1 -> 1}
\deftwocell{w : 1 -> 1} 

\title{Termination of Rewriting on Reversible Boolean Circuits as a Free $3$-Category Problem}
\author{Adriano Barile}[orcid=0000-0002-5122-0340, email=adriano.barile@unito.it]
\author{Stefano Berardi}[orcid=0000-0001-5427-0020, email=stefano.berardi@unito.it]
\author{Luca Roversi}[orcid=0000-0002-1871-6109, email=luca.roversi@unito.it]
\address{University of Turin, Italy}

\copyrightyear{2023}
\copyrightclause{Copyright for this paper by its authors.
	Use permitted under Creative Commons License Attribution 4.0
	International (CC BY 4.0).}

\conference{Proceedings of the 24th Italian Conference on Theoretical Computer Science, Palermo, Italy, September 13-15, 2023}

\begin{abstract}
Reversible Boolean Circuits are an interesting computational model under many aspects and in different fields,
ranging from Reversible Computing to Quantum Computing.
Our contribute is to describe a specific class of Reversible Boolean Circuits - which is as expressive as classical circuits -
as a bi-dimensional diagrammatic programming language.
We uniformly represent the Reversible Boolean Circuits we focus on as a free $3$-category $\ToffoliCategory$. This formalism allows us to incorporate the representation of circuits and of rewriting rules on them, and to prove termination of rewriting. Termination follows
from defining a non-identities-preserving functor from our free $3$-category $\ToffoliCategory$ into a suitable $3$-category $\Move$ that traces the
``moves" applied to wires inside circuits.
\end{abstract}


\maketitle

\section{Introduction}
\label{section-introduction}

The class of Reversible Boolean circuits (from now on, \textsf{RBC}) constitutes an interesting computational model, for many reasons. We name just some of them: once implemented, they may help to reduce electronic devices energy consumption \cite{landauer1961ibm}, easing miniaturization, due to a limited heat dissipation; they are at the core of cryptographic block cyphers analysis \cite{taborsky2018}, and of quantum circuits synthesis \cite{barenco1995,saeedimarkov2013}. Moreover, \emph{reversibility} means that if we execute the circuits in the opposite direction, e.g. bottom-up instead of top-down, we are able to recover the input. \textsf{RBC} can nevertheless simulate all non-reversible classical circuits \cite{toffoli80lncs}.

\begin{figure}[h]
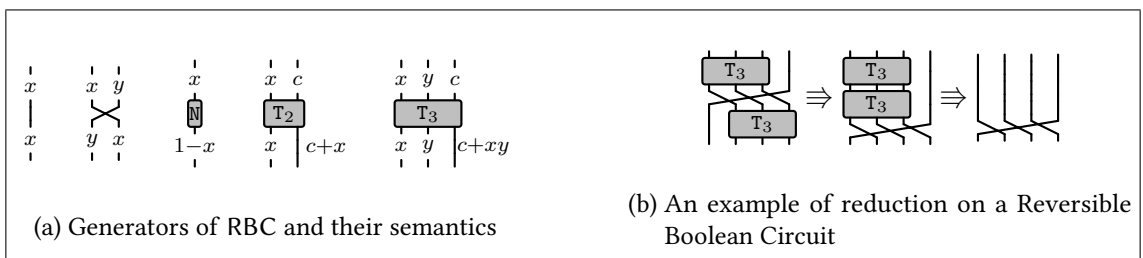

    \centering
    \begin{subfigure}[h]{0.45\textwidth}
        \centering
        \[
        \twocell{wx *1 w *1 wx}
        \quad
        \twocell{(wx *0 wy) *1 sw *1 (wy *0 wx)}
        \quad
        \twocell{wx *1 not *1 oN}
        \quad
        \twocell{(wx *0 wc) *1 t2 *1 (wx *0 oT2)}
        \quad
        \twocell{(wx *0 wy *0 wc) *1 t3 *1 (wx *0 wy *0 oT3)}
        \]
        \caption{Generators of \textsf{RBC} and their semantics}
        \label{fig:Generators and their semantics}
    \end{subfigure}
    \hfill
    \begin{subfigure}[h]{0.45\textwidth}
        \centering
        \[
        \twocell{(t3 *0 1)
                 *1
                 lad2
                 *1
                 (1 *0 t3)}
        \Rrightarrow
        \twocell{(t3 *0 1)
                 *1
                 (t3 *0 1)
                 *1
                 lad2}
        \Rrightarrow
        \twocell{ (w *0 w *0 w *0 w)
                  *1
                  (w *0 w *0 w *0 w)
                  *1
                  (w *0 w *0 w *0 w)
                  *1
                  lad2}
         \]
        \caption{An example of reduction on a Reversible Boolean Circuit}
        \label{fig:RevCircEx}
    \end{subfigure}
    \caption{Generators of Reversible Boolean Circuits and a reduction on them}
    \label{fig:Rev circ examples}
\end{figure}

\paragraph{Our focus} is to study a class of \textsf{RBC} in the lines of \cite{lafont2003towards}, i.e. as a bi-dimensional diagrammatic formal language built by series and parallel composition of the generators; the diagrams can be rewritten by rewriting rules preserving generator interpretation, as given in Fig.~\ref{fig:Generators and their semantics}. The inputs of the generators are on top. We assume $x, y, c \in \{0, 1\}$, and from left to right we have the gates: ``\emph{Identity}'', ``\emph{Swap}'', ``\emph{Negation}''
 (which we also call ``\emph{Toffoli-one}''), ``\emph{Toffoli-two}'', and ``\emph{Toffoli-three}''. We denote the boolean XOR operation with $+$; the AND is by juxtaposition.
Toffoli-two performs a \emph{controlled negation} ($\mathtt{CNOT}$): when the \emph{controller} $c$ is set to $1$, $x$ is negated, and is unchanged otherwise. The input $x$ is always carried over as an output.
Toffoli-three can be used to build conjuctions: if $c=0$, $(x,y,0) \mapsto (x,y,xy)$.

All classical boolean functions can be expressed in \textsf{RBC}, provided we add extra inputs/outputs in order to make the boolean functions invertible: in a reversible setting, the number of ouputs must be always equal to that of inputs.

All gates depicted in Fig.~\ref{fig:Generators and their semantics} represent invertible functions; moreover, they are self-inverse, i.e., if they yield an identity when applied twice.


\paragraph{Our goal} 
in the long term is  to answer specific domain questions related to the reversible language purposes.
Such questions in the short term require to answer more classical questions about rewriting.

An example of specific reversibility question is: ``Can we partition \textsf{RBC} into equivalence classes to find those containing the most efficient circuits, according to measures which depend on ancillae or the number of generators involved?''. We recall that ``ancillae'' is a technical notion, to identify variables used as temporary storage in a circuit $C\in $ \textsf{RBC}. Ancillae allow $C$ to compute the desired function while preserving the possibility of reverting the computation. A good reference to frame the role of ancillae in various contexts can be \cite{perumalla2013book}.

Another specific reversibility question can be: ``How can we decide the equivalence of reversible circuits that we obtain by compilers which translate classical boolean circuits into reversible ones, driven by some heuristic aimed at optimizing specific parameters of the output?''. The heuristic in \cite{Roetteler19Pebble} is just an example, of the many proposed in the literature, whose purpose, for instance, is to obtain good translations minimizing the quantity of ancillae.

Answering questions like the two above requires to formally and unambiguously know the representatives of equivalence classes in \textsf{RBC}. This amounts to asking basic questions about rewriting, such as: ``What is and how do we get normal forms of \textsf{RBC}?''

\paragraph{Contributions.}
A natural strategy to answer the last question is to look for normal forms w.r.t. an
interpretation-preserving rewriting on the circuits of \textsf{RBC}. Fig.~\ref{fig:RevCircEx} suggests what we mean by a simple example, in which we move a Toffoli-three gate next to another Toffoli-three so that they annihilate each other, because Toffoli-three are self-inverse operators. All \emph{base} reversible gates are self-inverse, while, in general, reversible circuits do not need to be self-inverse.

Rewriting bi-dimensional diagrammatic formal languages has well-know difficulties, mainly concerning their relations with the usual (linear) syntaxes expressing the same problem.

We here explore the use of a free $3$-category to represent both the reversible circuits and the rewriting rules on them. We build over Burroni's work \cite{burroni1993higher}, who introduced the notion of polygraphs to express algebraic theories and their reductions. Lafont \cite{lafont2003towards} used bidimensional diagrammatic syntax on product categories and an informal measure on them to prove termination results on several classes of boolean functions; Guiraud used $3$-polygraphs and diagrammatic reasoning to represent rewriting of circuit-like objects in \cite{guiraud2009polygraphic,guiraud2009finite}, together with a formal method to build terminating measures.
We focus on freely generated $3$-categories instead of polygraphs, to have access to the rewriting paths as $3$-morphisms; we call ``\emph{Toffoli $3$-category}'' the free $3$-category we introduce, and we denote it with $\ToffoliCategory$.

We remind the reader that an \emph{Abstract Rewriting System} (ARS) is said to be \emph{terminating} when there are no infinite chains of subsequent reductions, i.e. all reductions eventually yield a (not necessarily unique) \emph{normal form}.
We show termination of the rewriting system in $\ToffoliCategory$: termination follows from defining a functor from $\ToffoliCategory$ into a suitable $3$-category $\Move$ of moves applied to wires inside circuits.

We remark that the functor we introduce simplifies the functorial and differential interpretations in \cite{GUIRAUD2006341,guiraud2009polygraphic,guiraud2009finite}: our approach does not require to identify both a non-increasing measure on diagrams, which recalls a current flowing,
and a strictly decreasing measure, connected to the ``current'', which recalls heat, generated by the ``current'' itself. We just need to identify a decreasing measure on a monoid of strings. The measure keeps track of the ``moves'' being applied to each individual wire, and assigns them a cost.

\paragraph{Plan of our work.}
We sketch the definition of $3$-categories (\S \ref{section-categories}),
then we provide a bi-dimensional graphic representation for $3$-categories representing circuits (\S \ref{section-bidimensional}).
Next, we formally define the free $3$-category $\ToffoliCategory$ (\S \ref{section-circuit}),
and a functor from $\ToffoliCategory$ to a $3$-category $\Move$ of movements through circuits (\S \ref{section-functor}),which we show to be a decreasing measure of reductions.
Eventually we prove termination of reductions on \textsf{RBC} (\S \ref{section-termination}).
The reduction rules we utilized are \emph{syntactic} in nature; the normal form under algebraic equivalences remains an open problem \cite{lafont2003towards}.\\

\section{$3$-Categories, $3$-Functors and Free $3$-Categories}
\label{section-categories}
In this section we will briefly go over the definition of $n$-categories. We provide enough details for the purposes of representing circuits and syntactic reductions on them (level $3$).
For a complete description of higher categories we refer to e.g. \cite{baez1997,cheng2004}.
\paragraph{Motivations for using categories.} The literature using product categories to describe circuits is extensive; we follow Burroni and Guiraud and use higher category theory to describe reduction systems on algebraic structures, with particular attention to circuits \cite{burroni1993higher,guiraud2009finite,lafont2003towards}. The ``multi-level'' structure of $n$-categories provides a suitable model for bi-dimensional objects such as circuits; the third level effortlessly captures the notion of a rewriting system on circuits. Moreover, the use of categories allows us to reason up to ``bureaucratic'' identities such as those invoked when shortening/lengthening the wires by adding/removing identities.

We give an anticipation on our $3$-categorial model.
\begin{itemize}
	\item The ``base'' level will contain a single token object $\ast$ representing the empty space between wires in a circuit.
	\item A first level will contain input and output wires, with a ``monoidal'' product to generate bundles of multiple wires. In our model, we will consider a single \emph{wire} $\twocell{1}$ as a formal cell $\ast \rightarrow \ast$ between two ``empty spaces'' $\ast$.
	\item A second level will contain circuits as morphisms between $n$ input wires and $m$ output wires (in our context $n=m$ because of reversibility), with appropriate series and parallel compositions.
	\item	A third level will contain syntactic rewriting rules, considered as morphisms between circuits. Reductions should preserve the number of input/output wires of a circuit. 
\end{itemize}
  
We have that each level consists of morphisms, named \emph{cells}, with domain and codomain (here \emph{source} and \emph{target}) objects at the adjacent lower level. Source/target of a reduction is a circuit, source/target of a circuit is a set of I/O wires, and source/target of a wire is the empty space $\ast$.
The \emph{$i$-compositions} will provide an unified entity that captures both concatenation of wires, series and parallel composition of circuits, as well as three different ways to compose reductions. Also, the properties of $0$-, $1$-, $2$-composition will capture equations between circuits and between reductions.
The levels up to the second one yield a description of circuits isomorphic to the product categories formalization.  

We now provide an equational definition for $n$-categories following \cite{street1987}.

\begin{definition}[$n$-categories]
A (strict) $n$-category $\setC$ contains the following.
\begin{itemize}
	\item  A list of sets $C_i$, $0 \le i \le n$, called \emph{levels}, whose elements are called \emph{$i$-cells}.
	\item The maps $\source_i$, called \emph{$i$-source}, and $\target_i$, called \emph{$i$-target}, that associate to each $j$-cell $x$ with $0 \le i<j \le n$ two $i$-cells $\source_i(x)$, $\target_i(x)$ we respectively call the $i$-source and the $i$-target of $x$;
	\item The $j$-cells $x\star_i y$, defined for all $j$-cells $x$, $y$ and indices $0 \le i<j \le n$ such that $\target_i(x) = \source_i(y)$, called the \emph{$i$-composition} of $x$ and $y$, with $k$-source/target given by:
	\begin{enumerate}
		\item $\source_k(x\star_i y)=\source_k(x)$ and $\target_k(x\star_i y)=\target_k(y)$ when $0\leq k\leq i$;
		\item
		$\source_k(x \star_i y) =
		\source_k(x) \star_i \source_k(y)$ and $\target_k(x \star_i y) =
		\target_k(x) \star_i \target_k(y)$ when $j > k>i$.
	\end{enumerate}
	The $i$-composition is denoted in diagrammatic order (left-to-right).
\end{itemize}
The data above define an $n$-category if they satisfy the following further conditions.
\begin{itemize}
	\item
	\textbf{Globularity.}
	\[
	\source_{i-2} (\source_{i-1}(x))  = \source_{i-2} (\target_{i-1}(x)),\qquad \target_{i-2}(\source_{i-1}(x))  = \target_{i-2} (\target_{i-1}(x))
	\]
	for all $i$-cells $x$ with $2\leq i\leq n$. Globularity means that all $i$-cells connect two $(i-1)$-cells with the same $(i-2)$-source and $(i-2)$-target.
	\item
	\textbf{Associativity} of each $\star_{i}$.
	\item
	\textbf{Local Units.} For all $i$-cells $A$ with $0 \le i<j \le n$, there exists an identity $j$-cell $\id_{i,j,A}$, or $\id_A$ for short, 
such that $\source_i(\id_{i,j,A})=\target_i(\id_{i,j,A})=A$, the lower index source/targets of $\id_{i,j,A}$are those of $A$,
and for all $j$-cells $f$ we have
\begin{itemize}
	\item if $A = \source_i(f)$ then $\id_{i,j,A} \star_i f = f$, and
	\item if $A = \target_i(f)$ then $ f \star_i \id_{i,j,A} = f$.
\end{itemize}
Any $i$-composition of $j$-identity is a $j$-identity, for all $0 \le i<j \le n$. 

\item
\textbf{Exchange Rule.}
\begin{align*}
	\label{exchangerule}
	(\alpha \star_j \beta) \star_i (\gamma \star_j \delta) = (\alpha \star_i \gamma) \star_j (\beta \star_i \delta)
\end{align*}
for $\alpha,\beta,\gamma,\delta$ $k$-cells such that the above compositions are defined
and all $0 \le i<j <k\le n$.
\end{itemize}
\end{definition}

In order to build the measure that proves termination of circuits reductions, we need the notion of $3$-functor which we define for the general case.

\begin{definition}[$n$-functor]
Let $\setC$, $\setD$ be $n$-categories. An $n$-functor $\varphi: \setC \rightarrow \setD$
is a map such that for all $0 \le i  \le n$, $\varphi$ sends $i$-cells of $\setC$ into $i$-cells of $\setD$,
and such that for all $i$-cells $f,g$ of $\setC$, for all $0 \le j < i\le n$, and for all $j$-cells $A$ we have:
\begin{enumerate}
\item
\textbf{Source/Target preservation}.
$\source_j(F(f)) = F(\source_j(f))$, $\target_j(F(f)) = F(\target_j(f))$
\item
\textbf{Identity preservation (for $i$-cells and $\star_j$)}.
$F(\id_{j,i,A}) = \id_{j,i,F(A)}$
\item
\textbf{Composition preservation}.
 $F(f \star_j g) = F(f) \star_j F(g)$
\end{enumerate}
\end{definition}


\subsection*{Free $n$-Categories}
In the next section, we will represent circuits and reductions on them with a free $3$-category. A  free $n$-category consists of all well-formed expressions for objects in an $n$-category that are generated by a set of names for cells, identities and compositions, modulo all the equations we have for $n$-categories.

\begin{definition}[Free $n$-Categories]
Let $\vec{G}_{n}$ be a \emph{signature}, i.e. a list of sets $\vec{G}_n=(G_0, \ldots, G_{n})$, with $G_i$ sets of names for $i$-cells for $1 \le i \le n$, equipped with two maps $\source, \target:G_{i}\rightarrow G_{i-1}$, $1\leq i \leq n$ such that the globularity requirement is met. We define the \emph{free $n$-category $\setC_n$ generated by the signature $\vec{G}_n$} by induction on $n$.
\\

The $0$-category $\setC_0$ is just the set $G_0$. Assume we have defined the free $(n-1)$-category $\setC_{n-1}$ generated by the signature $\vec{G}_{n-1}$ as the $(n-1)$-category with levels $C_0,C_1,\dots,C_{n-1}$. 

\begin{itemize}
\item
An \emph{$n$-generator} of $C_n$ 
is any name $f \in \setC_{n-1}$ such that $n=1$ or $n \ge 2$ and $f$ satisfies the Globularity condition: 
$\source_{n-2}(\source_{n-1}(f)) = \source_{n-2}(\target_{n-1}(f))$ and $\target_{n-2}(\source_{n-1}(f)) = \target_{n-2}(\target_{n-1}(f))$. 
\item
Let $E_n$ be the set of \emph{$n$-constants} of $C_n$, containing all $n$-generators and all expressions $\id_{i,n,A}$ denoting the $i$-identity on $A$, for $0 \le i \le n-1$, $A \in C_i$.
\item
Let $E_n^*$ be the smallest set containing $E_n$ and all expressions $f \star_i g$ such that $f,g \in E^*_n$, $0 \le i \le n-1$ and $\target_{i}(f) = \source_{i}(g)$. Source/target maps are defined on $E^*_n$ by the source/target equations for $n$-categories.
\end{itemize}

$\setC_{n}$ is defined as the $n$-category with levels $C_{0},\dots,C_{n-1},C_{n}$ and $C_{n}=E^*_{n}/\sim$, where $\sim$ is the smallest equivalence relation compatible with $\star_0, \ldots, \star_{n-1}$ and including Associativity, Local Units and Exchange.
\end{definition}

Alternatively, the readers who are familiar with the definition of polygraphs \cite{burroni1993higher,guiraud2009polygraphic,GUIRAUD2006341} will notice that the free $n$-category generated by the signature $G$ is isomorphic to the $n$-category generated by the corresponding $n$-polygraph, which itself lacks categorial structure at the $n$-th level \cite{metayer2008}.


\section{A Bi-dimensional Diagrammatic Syntax for $3$-Categories}
\label{section-bidimensional}

In this section we restrict to free $3$-categories with a single $0$-cell $\ast$ and a unique generator for $1$-cell, i.e. the wire $\twocell{1}$, and we describe them as formal circuits and circuit reductions, by specifying what role the $i$-compositions take in the context of circuits.
We introduce a diagrammatic syntax that represents such free $3$-categories and the associated circuits. A great incentive to use a bi-dimensional syntax is that diagrams actually \emph{look like} circuits. We stress the fact that the categorial setting together with a diagrammatic formalism give a thorough, compact and sound presentation for circuital theories.
We invite the reader to think of $i$-compositions as a gluing operator of two objects along their common $i$-target and source, respectively.

\begin{definition}[Diagrams for a free $3$-category]
\label{def:Diagrams for a generic 3-category}
Let us assume that $G$ is a free $3$-category with generator sets  $\vec{G}=(G_{0},G_{1},G_{2})$, source and target maps $\source_i, \target_i$, a single $0$-cell and a unique generator for $1$-cells. The diagrammatic representation of $G$ is as follows.

\begin{itemize}
	\item \textbf{Generators.}
	\begin{itemize}
		\item $G_0 = \{\ast\}$ consists of a unique $0$-cell, representing a separator between input/output wires.
		We depict $\ast$ as a white area in the sheet of paper the diagram is drawn on.

		\item $G_1 = \left\{ \: \twocell{1} \: \right\}$ consists of a unique $1$-cell, which we call a \emph{wire}.
		A wire is a formal cell between the two portions of sheet that it divides, both marked with $\ast$.

		\item
		Each gate in $G_2$ is depicted as a circuit-like box with input and output wires representing $1$-source and $1$-target of the circuit.
		\[
		G_2 = \left\{
		\twocell{1 *1 ex1 *1 1} \:, \ldots,
		\twocell{2 *1 ex2  *1 2} \:, \ldots,
		\twocell{3 *1 ex3  *1 3} \:, \ldots,
		\twocell{n *1 gen *1 n} \:,\ldots \right\}.
		\]
		 If $g$ is a gate with $n$ input and output wires, we write $g:n\Rightarrow n$. The diagrams depicting gates preserve the meaning of $2$-cells as formal maps between $1$-cells (wires).
		\item
		$3$-cells in $G_3$, or reductions, are written as $\twocell{n *1 f *1 n} \Rrightarrow \twocell{n *1 g *1 n}$.
		We include no diagram for $3$-cells: this would involve $3$-dimensional objects \cite{guiraud2009polygraphic}.
	\end{itemize}
	\item \textbf{$i$-cells.}
	\begin{itemize}
		\item
		The set $G^*_0$ of all $0$-cells is again $G_0=\{\ast\}$.
		\item
		The set $G^*_1$ of all $1$-cells consists of all possible 	$0$-compositions $\star_0$ between $1$-cell generators.
		Any element of $G^*_1$ has the form
		\,$\twocell{1} \ast\ \twocell{1}\ \ast \twocell{1}\dots \ast \twocell{1}$ or simply
		\,$\twocell{1} \quad \twocell{1} \dots \twocell{1}$, freely generated as $0$-composition along their common $\ast$ white area.

		\item
		The set $G^*_2$ of all $2$-cells represents circuits, and it is obtained from gates in $G_2$ by closure w.r.t. $0$-composition $\star_0$ (parallel composition, or composition along a common $\ast$ area), and $1$-composition $\star_1$ (sequential composition, or composition along a common bundle of wires).

		\item \textbf{Compositions of $2$-cells.}
		\begin{itemize}
		\item Every two $2$-cells are $0$-composable, since they have as inputs and outputs $1$-cells with the same source and target
		(the only $0$-cell $\ast$). The $0$-composition of $2$-cells is depicted by putting the circuits next to each other along their common $\ast$ white area: $\twocell{(n *1 f *1 n) *0 (n *1 g *1 n)}$. This operation corresponds to the usual parallel composition and is read left-to-right.
		\item
		$1$-composition of $2$-cells is defined for pairs of $2$-cells such that $1$-target of the first one (its output wires) is equal to the $1$-source of the second one (its input wires). $1$-composition vertically stacks circuits by connecting common I/O set of wires: $\twocell{n *1 f *1 n *1 g *1 n}$. This operation corresponds to series composition and is read top-to-bottom.
	\end{itemize}

	\item \textbf{Circuit equivalences.}
\begin{itemize}

\item
The degenerate $1$-composition $\id_{A} \star_{1}f$ 
corresponds to prolonging the wires of $A = \source_{1}(f)$. The Unit rule $\twocell{n *1 (w *0 w) *1 2 *1 2 *1 f *1 n}=\twocell{n *1 2 *1 2 *1 f *1 n}=\twocell{n *1 2 *1 2 *1 f *1 (w *0 w) *1 n}$ graphically shows that circuits are defined independently from the length of wires. 
\item

The Exchange rule $ (f\star_{1} g)\star_{0}(h\star_{1}k)=	(f\star_{0} h)\star_{1}(g\star_{0}k)$ says that the diagram
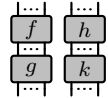
\begin{align*}
    \twocell{((n *1 f) *1 (n *1 g *1 n)) *0 ((n *1 h) *1 (n *1 k *1 n))}
\end{align*}
defines a single circuit which we can be equivalently read left-to-right, or top-to-bottom.
\end{itemize}
\end{itemize}

\item \textbf{3-cells}, or circuit reductions, can be composed with $\star_{0},\star_{1},\star_{2}$.
If $\alpha:f\Rrightarrow g$ and $\beta:h\Rrightarrow k$, then the compositions behave as follows.
	\[
	\begin{array}{ccccc}
		\twocell{(n *1 f *1 n) *0 (n *1 h *1 n)} \Rrightarrow \twocell{(n *1 g *1 n) *0 (n *1 k *1 n)} & \qquad &
		\twocell{n *1 f *1 n *1 h *1 n} \Rrightarrow \twocell{n *1 g *1 n *1 k *1 n} & \qquad &
		\twocell{n *1 f *1 n}\Rrightarrow\twocell{n *1 k *1 n} \\
		 & & & & \\
		\alpha\star_{0}\beta & \qquad & \alpha\star_1\beta & \qquad & \alpha\star_2\beta \\
		 & \qquad & (f,h \text{ and } g,k\text{ $1$-composable)} & \qquad & (g=h)
	\end{array}
	\]
\end{itemize}
\end{definition}

\begin{remark}
The diagrammatic notation, inspired by the string diagrams of categories, can be treated as a full $2$-dimensional syntax for product categories, as explained in \cite{acclavio2016phd}. Again in \cite{acclavio2016phd} there is a termination result for the rewriting system where associativity, local units and the exchange rule are understood as rewriting rules and not as equivalences,
and the concepts of \emph{squeezed form} and \emph{longest normal form} are introduced.
In this paper, we shall limit ourselves to reasoning on diagrams modulo the above equivalences.
\end{remark}

\section{The Free $3$-Category of Reversible Boolean Circuits}
\label{section-circuit}
In this section we will describe, as a free $3$-category $\ToffoliCategory$, both all reversible circuits we can obtain from the generators $\mathtt{SWAP}$, $\mathtt{NOT}$, $\mathtt{CNOT}$ (or $\mathtt{T}_{2}$, for Toffoli two), and $\mathtt{T}_{3}$ (or $\mathtt{CCNOT}$), and a particular set of reductions on such circuits. This set of gates is proven to be universal (with ancillae) in \cite{toffoli80lncs}.
The $3$-category $\ToffoliCategory$ includes, as its $3$-generators, a reduction set composed of rules that replace each pair of consecutive $\mathtt{SWAP}$ or consecutive Toffoli gates by an identity, implementing involutivity of reversible gates. We also arrange circuits in a canonical form, with a ``leaning'' to moving a $\mathtt{SWAP}$ down, and when this is not possible, to the right.
All reductions preserve the associated boolean function.
\begin{definition}[The free $3$-category $\ToffoliCategory$ of Reversible Circuits]
	The free $3$-category of Reversible Circuits is the free $3$-category specified by the following sets of $0$-, $1$-, $2$- and $3$-generators.
	\begin{itemize}
		\item
$R_{0}=\{\ast\}$,
$R_{1}=\{\:\twocell{1}\:\}$, where $\twocell{1}:\ast\rightarrow\ast$

\item
$R_{2}$ contains the following generators for reversible circuits, in this order: $\mathtt{SWAP}$, $\mathtt{NOT}$, $\mathtt{T}_{2}$, $\mathtt{T}_{3}$
\begin{align*}
R_{2}=
\left\{\:
\twocell{sw}:2\Rightarrow 2,\:
\twocell{not}:1\Rightarrow 1,\;
\twocell{t2}:2\Rightarrow 2,\:
\twocell{t3}:3\Rightarrow 3\:
\right\}
\end{align*}

		\item
$R_{3}=R_{p}\cup R_{a}\cup R_{s} \cup R_{t}$, where $R_{p}$ contains the following \emph{permutation rules}:
		\begin{align*}
			R_{p}=\left\{\:\twocell{sw *1 sw}\Rrightarrow \twocell{2}\:,\:\twocell{(sw *0 1) *1 (1 *0 sw) *1 (sw *0 1)}\Rrightarrow\twocell{(1 *0 sw) *1 (sw *0 1) *1 (1 *0 sw)} \:\right\},
		\end{align*}
		$R_{a}$ contains the following \emph{annihilation rules}:
		\begin{align*}
			R_{a}=\left\{\:\twocell{not *1 not}\Rrightarrow\twocell{1}\:,\:\twocell{t2 *1 t2}\Rrightarrow\twocell{2}\:,\:\twocell{t3 *1 t3}\Rrightarrow\twocell{3}\:\right\}
		\end{align*}

		$R_{s}$ contains the following \emph{sliding rules}:

		\[
		R_{s} = \left\{\:
		\begin{array}{ccc}
			\:\twocell{sw *1 (not *0 1)}\Rrightarrow\twocell{(1 *0 not) *1 sw}
\quad &,&\quad
\twocell{sw *1 (1 *0 not)}\Rrightarrow\twocell{(not *0 1) *1 sw}
\\
			\\
\twocell{lad1' *1 (t2 *0 1)}\Rrightarrow\twocell{(1 *0 t2) *1 lad1'}
\quad &,&\quad
			\twocell{lad1 *1 (1 *0 t2)}\Rrightarrow\twocell{(t2 *0 1) *1 lad1}
\\
			\\
\twocell{lad2' *1 (t3 *0 1)}\Rrightarrow\twocell{(1 *0 t3) *1 lad2'}

\quad &,&\quad
			\twocell{lad2 *1 (1 *0 t3)}\Rrightarrow\twocell{(t3 *0 1) *1 lad2}
		\end{array}
		\:\right\}
		\]

      and $R_{t}$ contains the following \emph{Swapped Toffoli rule}:
\begin{align*}
	R_{t}=\left\{ \:\twocell{(sw *0 1) *1 t3} \Rrightarrow \twocell{t3 *1 (sw *0 1)}\:\right\}
\end{align*}
	\end{itemize}
\end{definition}


\section{The Interpreting $3$-Functor}
\label{section-functor}
In this section, we will define a $3$-category $\Move$ of strings of ``moves'', each of them corresponding to the action of a gate on a single wire. Moves are elements of an ordered monoid that expresses the cost of series of singular moves along the wires of the circuit.
The circuits are then measured by a componentwise order on all moves on all wires. 
The reduced circuit should describe an equivalent boolean function with a reduced total cost. The interpretation of circuits into moves will be given as a $3$-functor from the free $3$-category $\ToffoliCategory$ of reversible circuits and a Toffoli base to $\Move$.

\begin{definition}[The ordered monoid $(\Monoid, <_{\Monoid})$ of moves]

\begin{enumerate}
\item
$\Monoid$ is the free monoid of words generated from the letters $\leftM, \rightM, \ToffoliM$.

\item
We order $w_1, w_2 \in \Monoid$ first by length, and when the lengths are the same, 
by the lexicographic order induced by the following order on letters: $\ToffoliM <_{\Monoid}\rightM <_{\Monoid} \leftM$.

\item
We write $<_{\Monoid}$ for  the order on $\Monoid$. 

\end{enumerate}
\end{definition}

The letters $\leftM, \rightM, \ToffoliM$ correspond to the three moves ``left-to-right", ``right-to-left" on the two wires of a $\mathtt{SWAP}$, and to the move ``Toffoli" on any wire of a Toffoli circuit. 
From $\Monoid$ we define a $3$-category $\Move$.

\begin{definition}[$\Move$]
$\Move$ is a collection of the following sets and ordered sets of $i$-cells $M_i$, $0\leq i\leq 3$.
\begin{itemize}
	\item $M_{0}=\{\ast\}$, a single element set.
	
	\item $(M_{1},<_1)$ is the set of all cartesian powers $\Monoid^n$ for $n \in \N$; the singleton $\Monoid^0$ is the identical $1$-cell, identified with the unique $0$-cell $\ast$.
	We define an order on $1$-cells as $\Monoid^n <_1 \Monoid^m$ if and only if $n<m$.
	We order vectors $\vec{w} \in \Monoid^n$ componentwise, with the \emph{product order} $<_{\Monoid^n}$ on $\Monoid^n$.
	
	\item $(M_2,<_2)$ is the set of all $<_{1}$-increasing maps $f:\Monoid^n \rightarrow \Monoid^n$ for some $n \in \N$ (all the identities $\id_{\Monoid^{n}}$ are included).
	Let $f,g:\Monoid^n \rightarrow \Monoid^n$ be $2$-cells with the same source and target.
	We define a \emph{strictly pointwise} order $f<_2 g$ on $2$-cells by: $f<_2g$ if and only if $f(x)<_{\Monoid^n} g(x)$ for all $x \in \Monoid^n$. We write $f \le_2 g$ for $f=g \vee f<_2 g$, that is:	either $f(\vec{x})=g(\vec{x})$ for all $\vec{x}$, or $f(\vec{x}) <_{\Monoid^n} g(\vec{x})$ for all $\vec{x} \in \Monoid^n$.
	
	\item $M_3$ contains all pairs  $r = \langle  f,g \rangle$ of $2$-cells with the same source and target $\Monoid^n$ for some $n$, such that $f \ge_2 g$. There is at most one $3$-cell between $f,g$, which merely signals the existence of a $\leq_2$-relation between the two functions. We think of each pair $\langle  f,g \rangle$ as a reduction from $f$ to $g$. The reduction is identical	if $f = g$, it is non-identical if $f <_2 g$.
\end{itemize}
The $i$-compositions on the $i$-cells are defined as follows.
\begin{itemize}
	\item \textbf{$0$-composition on $1$-cells} is $\Monoid^n\star_0 \Monoid^m=\Monoid^{n+m}$.
	\item \textbf{$0$-composition on $2$-cells} is the cartesian product of maps.
	\item \textbf{$1$-composition on $2$-cells} is the usual (sequential) composition of maps.
	\item \textbf{$0$-composition on $3$-cells} is defined as $\langle  f,g \rangle \star_0 \langle  f',g' \rangle = \langle  f \times f', g \times g' \rangle$.
	\item \textbf{$1$-composition on $3$-cells} is defined as $\langle  f,g \rangle \star_1 \langle  f',g' \rangle = \langle  f'f, g'g \rangle$.
	\item \textbf{$2$-composition on $3$-cells} is defined as  $\langle  f,g \rangle \star_2 \langle  g,h \rangle = \langle f,h\rangle$.
\end{itemize}
\end{definition}

\begin{remark}
	The order $<_{\Monoid}$ on words is well-founded, and in fact $(\Monoid,<_{\Monoid})$ is order-isomorphic to the order $(\N,<)$ on natural numbers.
	The order $<_2$ on maps is well-founded. In fact, each decreasing sequence $f >_2 f' >_2 f'' >_2 \ldots$ defines a decreasing sequence $f(\vec{\epsilon}) >_{\Monoid^n} f'(\vec{\epsilon}) >_{\Monoid^n} f"(\vec{\epsilon}) >_{\Monoid^n} \ldots$,
	where $\vec{\epsilon}=(\epsilon,\epsilon,\dots,\epsilon)$ and $\epsilon$ is the monoid unit (empty word).
	Therefore, every decresing sequence
	from $f$ has length at most the number of words which are less than $f(\vec{\epsilon})$ in $\Monoid^n$.
	
\end{remark}


\begin{proposition}
	$\Move$ is a $3$-category.
\end{proposition}


\begin{proof}
We first prove that $\star_0$, $\star_1$ are increasing on $2$-cells of $\Move$.
Assume $f \le_2 f'$ and $g\le_2 g'$, and either $f <_2 f'$ or $g <_2 g'$.
Then $\star_0$ is increasing: for all $x \in \Monoid^{n} = \source_1 f $, $y \in  \Monoid^{m} =\source_1 g$,
we have $(f \times g)(x,y)=(f(x),g(y)) <_{\Monoid^{n+m}} (f'(x),g'(y)) =  (f' \times g')(x,y)$.
 $\star_1$ is increasing: if $\target_1(f) = \Monoid^n = \source_1(g)$, then for all $x \in \source_1(f)$ we have
$(g f)(x)=g(f(x)) \le_{\Monoid^n} g(f'(x)) \le_{\Monoid^n} g'(f'(x)$,
and one of the two inequalities is strict, therefore  $g(f(x)) <_{\Monoid^n} g'(f'(x)$.
We then have also that $\star_0$ and $\star_1$ are increasing on $3$-cells of $\Move$. As a consequence, $i$-cells are closed w.r.t. $j$-compositions. Associativity, unit and exchange axioms are straightforward.
\end{proof}

\begin{definition}
A $3$-functor $\varphi: \ToffoliCategory \rightarrow \Move$
is \emph{strict in a $3$-cell $\alpha$} if
\begin{center}
 $\alpha$ is an identity $3$-cell in $\ToffoliCategory \Longleftrightarrow \varphi(\alpha)$ is an identity $3$-cell in $\Move$.
\end{center}
$\varphi$ is \emph{strict} if $\varphi$ is strict on all $3$-cells $\alpha$ of $\ToffoliCategory$.
\end{definition}

\begin{proposition}
Assume $\varphi: \ToffoliCategory \rightarrow \Move$ is a $3$-functor which is strict on all generators of $\alpha$.
Then $\varphi$ is strict.
\end{proposition}

\begin{proof}
By induction on $\alpha$, using the fact that $\varphi$ is a $3$-functor and $0$-, $1$-, $2$- and $3$-composition are increasing.
\end{proof}

We define a $3$-functor $\varphi: \ToffoliCategory \rightarrow \Move$ which is strict on all generators for $0$-, $1$- ,$2$-cells. Recall that a word is written left-to-right, the last letter being the last move.

\begin{definition}[Move interpretation]
\label{definition-functor}
	We define a map $\varphi$ by the following assignments on the generator gates of $\ToffoliCategory$.
	\begin{gather*}
          \varphi(\ast)=\ast \qquad
		\varphi(\:\twocell{1}\:)=\Monoid  \qquad \varphi(\:\twocell{sw}\:)(v,w)=(w \leftM, v \rightM) \\
		\varphi(\:\twocell{not}\:)(v)=v \ToffoliM \qquad \varphi(\:\twocell{t2}\:)(v, w)=(v \ToffoliM,w \ToffoliM)\\
		\varphi(\:\twocell{t3}\:)(v,w,z)=(v \ToffoliM, w \ToffoliM, z \ToffoliM) \qquad
          \varphi(r) = \langle \varphi(\source_2 r), \varphi(\target_2 r) \rangle
	\end{gather*}
\end{definition}
\begin{lemma}
	The map $\varphi$ of Def. \ref{definition-functor} extends in a unique way to a $3$-functor
	$\ToffoliCategory \rightarrow \Move$.
\end{lemma}
\begin{lemma}[Termination Lemma]
\label{lemma-termination}
	If the free $3$-category $\ToffoliCategory$ has a $3$-functor to $\Move$ which is strict on all generators for $3$-cells, then
all chains of non-identical $3$-cells in $\ToffoliCategory$ terminate.
\end{lemma}

$\ToffoliCategory$ is freely generated from a free $3$-category, therefore $\varphi$
is entirely and uniquely defined by the assignments on the generators of $\ToffoliCategory$, provided we recursively
check that the $3$-functor $\varphi$ preserves sources and targets, and sources and targets of sources.


\section{A Termination Result for Reversible Boolean Circuits}
\label{section-termination}

\begin{theorem}
	\label{theorem: termination1}
	The free $3$-category of Reversible Boolean Circuits terminates (all reduction sequences are finite).
\end{theorem}

\begin{proof}[of Termination for Reversible Boolean Circuits with the Toffoli base]
By Lemma \ref{lemma-termination}, we have to prove that $\varphi$ is strict on all generators for $3$-cells.
By definition of strictness, it is enough to prove that for all $\alpha\in R_{3}$ in $\ToffoliCategory$
we have that $\varphi(\alpha)$ is not an identity. 
By definition, $\varphi(\alpha)$ is equal to the pair $\langle \varphi(\source_{2}\alpha), \varphi(\target_2\alpha) \rangle$.
$\varphi(\alpha)$ is not an identity if $\varphi(\source_{2}\alpha)  >_2 \varphi(\target_2\alpha)$.
Suppose $\source_1(\alpha) = \Monoid^n$. By definition of the pointwise order on $2$-cells of $\Move$, we have to prove that
 $\varphi(\source_{2}\alpha)(v_1, \ldots, v_n) >_{\Monoid^n} \varphi(\target_2\alpha)(v_1, \ldots, v_n)$
 for all $ v_1, \ldots, v_n \in \Monoid$.


\begin{itemize}
\item
\textbf{Permutation rules.}
		\[
		\begin{array}{cccc}
			\varphi\left(\:\twocell{(sw *0 1) *1 (1 *0 sw) *1 (sw *0 1)}\:\right)(v,w,z) & & \varphi\left(\:\twocell{(1 *0 sw) *1 (sw *0 1) *1 (1 *0 sw)}\:\right)(v,w,z)\\
			\shortparallel & & \shortparallel\\
			(z \leftM \leftM, w \leftM \rightM , v \rightM \rightM )  & >_{\Monoid^n} & (z \leftM \leftM, w \rightM \leftM , v \rightM \rightM )
		\end{array}
		\]
The thesis follows from $\leftM\rightM > \rightM\leftM$.

		\item
\textbf{Annihilation rules.}
		\[
		\begin{array}{cccc}
			\varphi\left(\:\twocell{t3 *1 t3}\:\right)(v,w,z) & & \varphi\left(\:\twocell{3}\:\right)(v,w,z)\\
			\shortparallel & & \shortparallel\\
			(v \ToffoliM \ToffoliM, w \ToffoliM \ToffoliM, z \ToffoliM \ToffoliM) & >_{\Monoid^n} & (v,w,z)
		\end{array}
		\]
 Words on the right-hand-size are all shorter and therefore smaller.

		\item \textbf{Left-to-Right Sliding rules.}

	The reduction moving the Toffoli gate upwards and to the right swaps the lowest letter of the monoid $\ToffoliM$, with a higher letter, $\leftM$, in all movements $v,w,z,t$ but $v$. Below, we consider the case of the circuit Toffoli $3$.
		\[
		\begin{array}{cccc}
			\varphi\left(\:\twocell{lad2' *1 (t3 *0 1)}\:\right)(v,w,z,t) & & \varphi\left(\:\twocell{(1 *0 t3) *1 lad2'}\:\right)(v,w,z,t)\\
			\shortparallel & & \shortparallel\\
		(w \leftM \ToffoliM,  z \leftM \ToffoliM, t \leftM \ToffoliM,  \ v \rightM \rightM \rightM)
			& >_{\Monoid^n} &
		(w \ToffoliM \leftM ,  z \ToffoliM \leftM , t  \ToffoliM \leftM, \  v \rightM \rightM \rightM)
		\end{array}
		\]
	The thesis follows from $\leftM  \ToffoliM >_\Monoid \ToffoliM \leftM $. The Right-to-Left Sliding rules are the mirror case.
	\item \textbf{Swapped Toffoli.}
	This case follows from $\leftM  \ToffoliM > \ToffoliM \leftM $ and $\rightM \ToffoliM > \ToffoliM \rightM$. 
	\end{itemize}
\end{proof}
\section{Conclusion}
\label{section:Conclusion}
This work explores $3$-categories as a unified formal framework for modeling two concepts. First, it examines Reversible Boolean Circuits, which are treated as diagrams generated from a base through the iterative application of series/parallel compositions. Then, it investigates the termination of a rewriting system on these circuits.
\textbf{Toff} is the free $3$-category which effectively formalizes circuits and rewriting rules.
\textbf{Move} is the $3$-category supplying the well founded-order in a monoid of strings.
Termination follows from interpreting \textbf{Toff} into \textbf{Move} by means of a (strict) functor $\varphi$ that intuitively represent traces of moves inside the circuit being rewritten, by exploiting the common $3$-category structure.
A possible extension of this work is to evaluate its effectiveness in proving the termination of reductions of free $3$-polygraphs generated from alternative bases which include Fredkin and Peres gates.
A second development must focus on confluence: rewriting in \textbf{Toff} is not. We claim that the following circuits has two non-confluent normal forms:

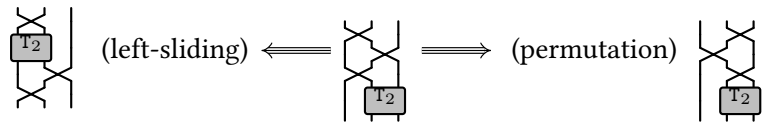
\begin{center}
	\begin{tikzcd}
		\twocell{(sw *0 1) *1 (t2 *0 1) *1 (1 *0 sw) *1 (sw *0 1)}    \ \ \    \mbox{ (left-sliding)}
		&
		\twocell{(sw *0 1) *1 (1 *0 sw) *1 (sw *0 1) *1 (1 *0 t2)} \arrow[l,triple] \arrow[r,triple]
		&
		\mbox{ (permutation)}   \ \ \    \twocell{(1 *0 sw) *1 (sw *0 1) *1 (1 *0 sw) *1 (1 *0 t2)}
		\enspace 
	\end{tikzcd}
\end{center}

We are pursuing various directions in order to obtain confluence. However, the problem at hand is known to be non-obvious. Guiraud warns that a 3-polygraph may generate an infinite number of critical pairs which, in specific cases, can be categorized into a finite set of patterns, eventually leading to confluence \cite{GUIRAUD2006341}.

\bibliography{bibliography}
\end{document}